\documentclass[a4paper,11pt]{article}

\usepackage{amsmath,amsthm,amssymb,xspace}
\usepackage{authblk}
\usepackage{tabularx,hyperref}
\usepackage{color,url}
\usepackage{enumitem}

\newlength{\figurewidth}
\newlength{\smallfigurewidth}

\long\def\ignore#1{\vskip 0pt}
\newcommand{\Oh}[1]{\ensuremath{\mathcal{O}\!\left({#1}\right)}}

\theoremstyle{plain}
\newtheorem{theorem}{Theorem}
\newtheorem{lemma}[theorem]{Lemma}
\newtheorem{corollary}[theorem]{Corollary}
\theoremstyle{definition}

\theoremstyle{remark}

\newcommand{\occ}{\mathsf{occ}}

\newcommand{\Tdelta}{T_\delta}

\newcommand{\gzip}{{\sf gzip}\xspace}
\newcommand{\xz}{{\sf xz}\xspace}
\newcommand{\scan}{{\sf scan}\xspace}
\newcommand{\fmo}{{\sf fm-oppm}\xspace}
\newcommand{\sso}{{\sf simd-oppm}\xspace}
\newcommand{\smi}{{\sf smi}\xspace}
\newcommand{\ssmi}{{\sf smi}\xspace}

\newcommand{\lsk}{{\sf lsk}\xspace}
\newcommand{\csa}{{\sf csa}\xspace}
\newcommand{\prcz}{{\bf prices}\xspace}
\newcommand{\ecgz}{{\bf ecg}\xspace}
\newcommand{\randwz}{{\bf rwalk}\xspace}

\def\rdown#1{\lfloor #1 \rfloor}

\newcolumntype{R}{>{\raggedleft\arraybackslash}X}

\begin{document}

\title{\textbf{A Compact Index for Order-Preserving Pattern
Matching}\thanks{A preliminary version of this paper appeared in the Proc. IEEE Data Compression Conference, DCC~2017, Snowbird, UT, USA, 2017~\protect{\cite{dcc/DecaroliGM17}}.}}

\author[1]{Gianni Decaroli}
\author[2]{Travis Gagie}
\author[1,3]{Giovanni Manzini}
\affil[1]{Computer Science Institute, University of Eastern Piedmont, Italy}
\affil[2]{Diego Portales University, Santiago, Chile}
\affil[3]{IIT-CNR, Pisa, Italy}

\date{}
\maketitle
\thispagestyle{empty}

\begin{abstract}

Order-preserving pattern matching was introduced recently but it has already attracted much attention. Given a reference sequence and a pattern, we want to locate all substrings of the reference sequence whose elements have the same relative order as the pattern elements. 
For this problem we consider the offline version in which we build an index for the reference sequence so that subsequent searches can be completed very efficiently. 
We propose a space-efficient index that works well in practice despite its lack of good
worst-case time bounds. Our solution is based on the new approach of
decomposing the indexed sequence into an {\em order component}, containing
ordering information, and a {\em $\delta$ component}, containing information
on the absolute values. Experiments show that this approach is viable, faster than the available alternatives, and it is the first one offering simultaneously {\em small space usage} and {\em
fast retrieval}.

\end{abstract}

\section{Introduction}

The problem of {Order-Preserving Pattern Matching} consists in finding,
inside a numerical sequence $T$, all subsequences whose elements are in a
given relative order. For example, if the pattern is $P=(1,2,3,4,5)$ we need
to find all increasing subsequences of length five; so if $T=(10, 20, 25, 30,
31, 50, 47,49)$ we have a first match starting with the value 10, a second
match starting with the value 20, and no others.

This problem is a natural generalization of the classic exact matching
problem where we search for subsequences whose values are exactly those of
the pattern. Order-preserving matching is useful to search for trends in time
series like stock market data, biomedical sensor data, meteorological data,
{\it etc.} In the last few years this
problem has received much attention. Not surprisingly, most of the results
are generalizations of algorithms and techniques used for exact matching.
In~\cite{isaac/BelazzouguiPRV13,ipl/ChoNPS15,tcs/KimEFHIPPT14,
ipl/KubicaKRRW13} the authors propose online solutions inspired by the classical
linear time Knuth-Morris-Pratt and Boyer-Moore
algorithms~\cite[Chap.~2]{Gusfield:97}. In~\cite{spire/CrochemoreIKKLPRRW13}
the authors consider the offline problem in which $T$ can be preprocessed and
propose an index that generalizes the classical Suffix Tree data
structure~\cite[Chap.~5]{Gusfield:97}. Finally
in~\cite{stringology/CantoneFK15,spire/ChhabraGT15,stringology/ChhabraKT15,ipl/ChhabraT16,
FaroK16,cpm/GawrychowskiU14} the authors consider approaches based on the
concept of filtration and seminumerical matching~\cite[Chap.~4]{Gusfield:97}.

In this paper we extend to Order-Preserving Matching another well known idea
of exact matching: simultaneously compressing and indexing a sequence of
values~\cite{csur/NavarroM07}. In Section~\ref{sec:algo} we show how to compactly represent a sequence
$T$ so that given a pattern $P$ we can efficiently report all subsequences
of $T$ whose elements are in the same relative order as the elements of $P$.
Our contribution is based on the new idea of decomposing the sequence $T$
into two components: the {\em order component} and the {\em $\delta$
component}. Informally, the order component stores the information about the
relative order of the elements of $T$ inside a window of a preassigned size,
while the $\delta$ component contains the information required for
reconstructing $T$ given the order component. The order component is stored
into a compressed suffix array while the $\delta$ component is stored using
an ad-hoc compression technique.

To search for a pattern we compute its ordering information and then we
search for it in the compressed suffix array of the order component. Since the
information in the order component is only partial, this search gives us a
list of potential candidates which are later verified using the $\delta$
component. In other words, the search in the compressed suffix array is a
sort of filtering phase that uses the index to quickly select a set of
candidates, discarding all other subsequences in $T$ that certainly do not
match. 


The overall efficiency of our approach depends on some parameters of the
algorithm whose influence will be experimentally analyzed in
Section~\ref{sec:experiments}. The bottom line is that our index takes
roughly the same space as the compressor \gzip\ and can report the order-preserving occurrences of a pattern an order of magnitude faster than simply scanning the input file. Indeed, in Section~\ref{sec:experiments} we show that our solution is significantly faster than the current fastest online algorithm~\cite{spe/TSOT16} and the offline algorithm in~\cite{stringology/ChhabraKT15} which is also based on the idea of building a ``partial'' index to speed up the search.

The complete source code for the index construction and search algorithms are available at the repository \url{https://gitlab.com/manzai/order-preserving-index}.

\section{Problem formulation and previous results}\label{sec:formulation}

Let $T[1,n]$ denote a sequence of $n=|T|$ numerical values. We write $T[i]$
to denote the $i$-th element and $T[j,k]$ to denote the subsequence
$T[j]T[j+1]\cdots T[k]$. Given two sequences $P$, $Q$, we say that they are
{\em order isomorphic}, and write $P\approx Q$, if $|P|=|Q|$ and the relative
order of $P$'s and $Q$'s elements is the same, that is
\begin{equation}\label{eq:oiso}
P[i] < P[j] \Longleftrightarrow Q[i] < Q[j]\qquad \mbox{for }
1 \leq i,j \leq |P|.
\end{equation}
Hence $(1,3,4,2)$ is order isomorphic to $(100,200,999,101)$ but not to
$(1,3,4,5)$. Given a reference sequence $T[1,n]$ and a pattern $P[1,p]$ the
{\em order-preserving pattern matching problem} consists in finding all
subsequences $T[i+1, i+p]$ such that $T[i+1,i+p]\approx P[1,p]$. In this
paper we consider the offline version of the problem in which the sequence
$T$ is given in advance and we are allowed to preprocess it in order to
speed up subsequent searches.

The first algorithms for the order-preserving pattern matching problem were designed
for the online case, where one is allowed to preprocess the pattern but not the text, 
and inspired by the classical Knuth-Morris-Pratt and Boyer-Moore
algorithms~\cite{isaac/BelazzouguiPRV13,ipl/ChoNPS15,tcs/KimEFHIPPT14,ipl/KubicaKRRW13}.
The proposed algorithms have guaranteed linear time worst-case complexity or
sublinear time average complexity. However, the best
results in practice are obtained by algorithms based on the concept of
filtration, in which some sort of ``order-preserving'' fingerprint is applied
to the text and the pattern~\cite{spire/ChhabraGT15,%
wea/ChhabraT14,ipl/ChhabraT16,FaroK16}. The practical performance of these algorithms can be further improved by exploiting the Single Instruction Multiple Data (SIMD) parallelism now commonly offered by modern processors through the SSE instruction set~\cite{stringology/CantoneFK15,spe/TSOT16}. Note that all the algorithms based on filtration and/or SIMD are superlinear in the worst case.

For the offline problem, Crochemore et al.~\cite{tcs/CrochemoreIKKLP16} showed
how, given a sequence $T[1,n]$, in $\Oh{n \log (n) / \log \log n}$ time we
can build an $\Oh{n \log n}$-bit index such that later, given a pattern
$P[1,p]$, we can return the starting positions of all the \textsf{occ}
order-preserving matches of $P$ in $T$ in optimal $\Oh{m + \mathsf{occ}}$
time. A more space-economical solution is proposed in~\cite{esa/GagieMV17}
consisting of an index taking $\Oh{n \log\log n}$ extra bits in addition to
the text. The one in~\cite{esa/GagieMV17} is the first index for
order-preserving matching that uses $o(n\log n)$ bits with guaranteed worst
case search bounds.  Its weaknesseses are that it can handle only patterns whose
length is less than a given bound that is polylogarithmic in $n$, and that it
returns the position of only one match (if there is one).

In this paper we are interested in practical approaches that work well in
practice even if they do not have competitive worst case bounds on the search
cost. In~\cite{stringology/ChhabraKT15} Chhabra et al. show how to speedup
search building an FM-index~\cite{isci/FerraginaM01,FM05} on the binary string
expressing whether in the input text each element is smaller or larger than
the next one. Our proposal can be seen as a generalization of this work:
instead of extracting a binary sequence from $T$ we extract information on
the relative order of the elements inside a sliding window of size $q$. In
addition, we also compute a {\it $\delta$ component} containing the
information not stored in the order component. As a result we obtain the first
compressed representation of a sequence that is simultaneously an index for
order-preserving matching. Both our solution and the one
in~\cite{stringology/ChhabraKT15} use the index to quickly select a set of
candidates, discarding all other subsequences in $T$ that certainly do not
match. However, not only we do use a more ``informative'' index but we also exploit the $\delta$ component to speed up the verification phase obtaining much better performance in practice.

\section{Data representation and search algorithm}\label{sec:algo}

\subsection{The ordering and delta component}

Given a window size $q>1$ and a sequence $T[1,n]$ we define its {\em order
component} $T_o[1,n]$ as follows. For $i=1,\ldots,n$ let $i_q =
\max(1,i-q+1)$ and define
\begin{equation}\label{eq:Todef=}
T_o[i] = \begin{cases}
0.5 & \mbox{if } T[i] < \min T[i_q,i-1] \mbox{ or } i=1\\
k & \mbox{if }
T[i-k] = \max_{i_q\leq j < i} \{T[j] \vert\;  T[j] \leq T[i]\}\,\mbox{ and }\,
T[i-k]=T[i]\\
k+0.5 & \mbox{if }
T[i-k] = \max_{i_q\leq j < i} \{T[j] \vert\; T[j] \leq T[i]\}\,\mbox{ and }\,
T[i-k] < T[i].
\end{cases}
\end{equation}
In other words: if $T[i]$ is the smallest element in $T[i_q,i]$ we set
$T_o[i]=0.5$; otherwise, if $T[i-k]$ is the immediate predecessor of $T[i]$
in $T[i_q,i-1]$ we set $T_o[i]=k$ if $T[i]=T[i-k]$, or $T_o[i]=k+0.5$ if
$T[i]> T[i-k]$. Note that if $T_o[i]\geq 1$ then $T[i-\rdown{T_o[i]}]$ is the
predecessor of $T[i]$ in $T[i_q,i-1]$. The values of $T_o$ belong to the set
$\{0.5,1,1.5,\ldots,q-0.5\}$ which has $2q-1$ elements overall. For example,
if $q=4$ for $T = (3,8,3,5,-2,9,6,6)$ it is $T_o =
(0.5,1.5,2,1.5,0.5,2.5,3.5,1)$.

We call $T_o$ the {order component} for $T$ since it encodes ordering
information for $T$'s elements within a size-$q$ window. Formally $T_o$
depends also on $q$ but for simplicity we omit it from the notation.
Obviously, if $P\approx Q$ then $P_o = Q_o$. However, we are interested in
finding the order-preserving occurrences of $P$ within a longer reference
sequence $T$ and we will make use of the following more general result.

\begin{lemma}\label{lemma:main}

Let $i$ be such that $P[1,p] \approx T[i+1,i+p]$. Then, if $P_o[j]$ is whole
it is $T_o[i+j] = P_o[j]$; if $P_o[j]$ is fractional it is $T_o[i+j] =
P_o[j]$ for every $j$ such that
\begin{equation}\label{eq:lemma:main}
2 \leq j \leq p \qquad \mbox{and}\qquad (j - \rdown{T_o[i+j]}) \geq 1.
\end{equation}
\end{lemma}

\begin{proof}

If $k = P_o[j]$ is whole, then by definition $P[j]$ is equal to $P[j-k]$ and
different from any character in $P[j-k+1,j-1]$. Since $P[1,p] \approx
T[i+1,i+p]$ we must similarly have $T[i+j]=T[i+j-k]$ and $T[i+j]$ different
from any character in $T[i+j-k+1,i+j-1]$. Hence $T_o[i+j] = k = P_o[j]$ as
claimed.

Let $P_o[j]$ be fractional and such that~\eqref{eq:lemma:main} is satisfied.
Let $w = \min(q,i+j)$ and $v=\min(q,j)$. Note that $w$ (resp. $v$) is the
size of the subsequence which is considered for determining $T_o[i+j]$ (resp.
$P_o[j]$). Clearly $w\geq v$.

If $T_o[i+j]=0.5$, then $T[i+j]$ is the smallest element in the subsequence
$T[i+j-w+1,i+j]$. A fortiori $T[i+j]$ is the smallest element in
$T[i+j-v+1,i+j]$. The hypothesis $P[1,p] \approx T[i+1,i+p]$ implies that
likewise $P[j]$ must be the smallest element in $P[j-v+1,j]$ so it must be
$P_o[j]=0.5$. Assume now $T_o[i+j]\geq 1 $ and let $\ell_j = j -
\rdown{T_o[i+j]}$. By construction $T[i+\ell_j]$ is the immediate predecessor
of $T[i+j]$ in the subsequence $T[i+j-w+1,i+j]$. The condition $j -
\rdown{T_o[i+j]} \geq 1$ implies $1 \leq \ell_j < j$. 

In other words, if $T_o [i + j]$ is fractional and the distance $\rdown{T_o [i + j]}$ it points back to $T [i + j]$'s predecessor in $T [i_q,i + j - 1]$ is less $j$, then that predecessor is in $T[i + 1, i + j - 1]$; therefore, since $T[i + 1,i + p] \approx P [1,p]$, the distance back from $P[j]$ to it's predecessor is the same. We conclude that $P_o [j] = T_o [i + j]$ as claimed.
\end{proof}


\begin{corollary}\label{cor:main}
If $P[1,p]\approx T[i+1,i+p]$ we must have
\begin{equation}\label{eq:main_cor}
T_o[i+2,i+p] = x_2\,  x_3\,  \cdots  x_{q-1}\,  P_o[q]\,  P_o[q+1]  \cdots  P_o[p]
\end{equation}
where for $j=2,\ldots,q-1$ it is:  $x_j = P_o[j]$ if $P_o[j]$ is whole, and
either $x_j = P_o[j]$ or $x_j\geq j$ if $P[j]$ is fractional.\hfill\qed 
\end{corollary}

\begin{proof}
Since $\rdown{T_o[\cdot]} \leq q-1$, by the above lemma if $P[1,p]\approx T[i+1,i+p]$ then for every $j$, $q \leq j \leq p$ it is $T_o[i+j] = P_o[j]$. 

For $j=2,3,\ldots,q-1$ the lemma establishes that $T_o[i+j]$ can be different from $P_o[j]$ only if $P_o[j]$ is fractional and $\rdown{T_o[i+j]}) > j-1$. The latter implies $T_o[i+j] \geq 1$ as claimed. 
\end{proof}

In view of the above corollary, our strategy to solve the order-preserving
matching problem is to build a compressed full text index for the sequence $T_o$. Then,
given a pattern $P[1,p]$ with $p>q$, we compute $P_o[1,p]$ and then the set
of positions $\{i_1, i_2, \ldots, i_m\}$ satisfying~\eqref{eq:main_cor}. Clearly,
we can have $P[1,p]\approx T[i+1,i+p]$ only if $i \in \{i_1,
i_2,\ldots,i_m\}$. Note that finding $\{i_1, i_2, \ldots, i_m\}$ will usually require  more than one distinct search. For example, if $q=4$ and $P=(3, 5, 2, 6, 5, 1, 5)$, then $P_o = (0.5, 1.5, 0.5, 2.5, 3, 0.5, 2)$. Then , the sequences to be sought in $T_o$ are 
$$
\begin{array}{ccc}
(1.5, 0.5, 2.5, 3, 0.5, 2),& (2.5, 0.5, 2.5, 3, 0.5, 2),&
(3.5, 0.5, 2.5, 3, 0.5, 2)\cr
(1.5, 3.5, 2.5, 3, 0.5, 2),& (2.5, 3.5, 2.5, 3, 0.5, 2),&
(3.5, 3.5, 2.5, 3, 0.5, 2).
\end{array}
$$

Clearly, any condition stated in terms of the ordering component only, like Corollary~\ref{cor:main}, can only be a necessary condition for an order-preserving match. Thus, we need to verify the candidate positions satisfying Corollary~\ref{cor:main} in a verification phase using the actual values of the sequence $T$. Since we are interested in indexing {\em and} compressing $T$, instead of simply storing $T$ we introduce a representation that takes advantage of the values in $T_o$ that are already stored in the index.

Given $T[1,n]$ and $T_o[1,n]$, we define a new
sequence $T_\delta[1,n]$ as follows. Let $T_\delta[1] = T[1]$. For
$i=2,\ldots,n$ let $i_q = \max(1,i-q+1)$ and:
\begin{equation}\label{eq:Tdeltadef}
T_\delta[i] =
\begin{cases}
\min T[i_q,i-1] - T[i] & \mbox{if } T_o[i] = 0.5\\
T[i] - T[i-\rdown{T_o[i]}] & \mbox{if } T_o[i]\geq 1.
\end{cases}
\end{equation}
Notice that for $i\geq 2$, $T_\delta[i] \geq 0$. Indeed, if $T_o[i]=0.5$ then
by~\eqref{eq:Todef=} $T[i] < \min T[i_q,i-1]$. If $T_o[i]\geq 1$, since
$T[i-\rdown{T_o[i]}]$ is the immediate predecessor of $T[i]$ in $T[i_q,i]$, it
is $T_\delta[i] \geq 0$. We call $T_\delta$ the {$\delta$ component} of
$T$. While $T_o$ provides information on the ordering of $T$'s elements,
$T_\delta$ contains information on their absolute values. It is straightforward to
verify that given $T_o$ and $T_\delta$ we can retrieve $T$ in linear time.

Summing up, our approach to compress and index a sequence $T$ and support
order-preserving pattern matching is the following: select a window size $q$
and build the components $T_o$ and~$T_\delta$. Then build a compressed index
for $T_o$ and compress $T_\delta$. These two components together constitute
our index. The above description is quite general and can be realized in
different ways. In the following we describe our particular implementation
and experimentally measure its effectiveness.

\subsection{Representation of the components $T_o$ and
$T_\delta$}\label{sec:compr}

We represent $T_o[1,n]$ using a Compressed Suffix Array (\csa) consisting of
a Huffman-shaped Wavelet Tree built on the BWT of the sequence~$T_o$. To this end we use the {\sf csa\_wt} class from the {\sf sdsl-lite} library~\cite{wea/GogBMP14} (in practice we consider the elements of $T_0$ multiplied by two to avoid non-integer values). Given any pattern $p$ the \csa\ can compute in
$\Oh{p}$ time the range of rows $[b,e]$ of the Suffix Array of $T_o$ prefixed
by $p$. To find the actual position in $T_o$ of each occurrence of $p$, the
\csa\ also stores the set of Suffix Array entries containing text positions which
are multiples of a given block size $B$. Then, for each row $r \in
[b,e]$ we move backward in the text using the LF-map until we reach a marked
Suffix Array entry from which we can derive the position in $T_o$ of the
occurrence that prefixes row $r$. The above scheme uses $\Oh{ n + (n/B)\log
n}$ bits of space and can find the position of all (exact) occurrences of $p$
in $T_o$ in $\Oh{|p|+ B\,\occ}$ time, where $\occ = e-b+1$ is the number of
occurrences. Clearly, the parameter $B$, chosen when we build the \csa, 
offers a trade-off between space usage and running time.

We do not use an index for $\Tdelta$. However, during the verification phase
we need to extract (decompress) the values in random portions of
$\Tdelta$. For this reason we split $\Tdelta$ into blocks of size $B$ (i.e., the
same size used for the blocks in the \csa\ of $T_o$) and we compress each
block independently. The $k$-th block consists of the subsequence
$\Tdelta[kB+1,kB+B]$, except for the last block which has size $(n \bmod B)$.
Additionally, we use a header storing the starting position of each block.
Hence, given a block index we can decompress it in $\Oh{B}$ time.

To compactly represent a block of $\Tdelta$ we take advantage of the fact
that the corresponding values in $T_o$ are available during compression and
decompression. Recalling the definition of $\Tdelta[i]$
in~\eqref{eq:Tdeltadef}, we partition the values in $T$ into three classes:
\begin{enumerate}[itemsep=-0.1cm,topsep=-0.1cm]
\item those such that $T[i] < \min T[i_q,i-1]$ are called {\it minimal};
\item those such that $T[i] > \max T[i_q,i-1]$ are called {\it maximal};
\item all other values are called {\it intermediate}.
\end{enumerate}\vspace{0.1cm}
The class of $T[i]$ can be determined by both compressor and
decompressor, the latter using $T_o[i]$, before it is (de)coded. For each block we define
$$
m  = \max \{\Tdelta[i] \;\vert\; i \mbox{ is minimal}\},\qquad\qquad
M  = \max \{\Tdelta[i] \;\vert\; i \mbox{ is maximal}\};
$$
and we store these two values at the beginning of the block. When we
encounter a minimal (resp. maximal) value $T[i]$ we know that the
corresponding value $\Tdelta[i]$ will be in the range $[1,m]$ (resp.
$[1,M]$). When we encounter an intermediate value $T[i]$ if $T_o[i]$ is an
integer then we know that $\Tdelta[i]=0$ and there is nothing to encode. If
$T_o[i]$ is fractional we know that $\Tdelta[i]$ will be in the range
$[1,v-T[i-\rdown{T_o[i]}]-1]$ where $v$ is the smallest element in
$T[i_q,i-1]$ larger than $T[i-\rdown{T_o[i]}]$.

Summing up, compressing a block of $\Tdelta$ amounts to compressing a
sequence of non-negative integers $\ell_1, \ell_2, \ldots, \ell_B$ with the
additional information that for each $\ell_i$ both encoder and decoder know
an upper bound $w_i\geq \ell_i$. We have tested several compressors for this
setting and we got the best results using the {\em log-skewed}
coder~\cite{spe/ManziniR04}. Such an encoder represents an integer $\ell \in
[0,w)$ using at most $\lceil\log_2(w)\rceil$ bits, but if $w$ is not a power
of two the smallest values in the range $[0,w)$ are encoded using fewer than
$\lceil\log_2(w)\rceil$ bits.

\subsection{Searching for a pattern}\label{sec:search}

Given the above representations of $T_o$ and $\Tdelta$, we compute the order-preserving occurrences of a pattern $P$ in $T$ as follows. First we compute
$P_o$ and locate in  $T_o$'s \csa\ the row ranges prefixed by each one of the
sequences satisfying Corollary~\ref{cor:main}. If the window size is $q$
there are at most $(q-1)!$ such sequences. Recall that the basic operation of
a \csa\ is the {\em backward search} in which, given the range of rows
prefixed by a substring $\alpha$ and a character $c$, we find in $\Oh{1}$
time the range of rows prefixed by $c\alpha$. This suggests we compute the
desired set of row ranges with a two-step procedure: first ({\sf Phase 1})
with $p-q+1$ backward search steps we compute the range of rows prefixed by
$P_o[q,p]$; then ({\sf Phase 2}) with additional backward search steps we
compute the range of rows prefixed by $x_2  x_3 \cdots x_{q-1} P_o[q]\cdots
P_o[p]$ for each $(q-2)$-tuple $x_2,\ldots,x_{q-1}$ satisfying the conditions
of Corollary~\ref{cor:main}. Phase 2 can require up to $q!$ backward search
steps, but the number of steps is also upper bounded by $q$ times the number
of row ranges obtained at the end of the phase, which is usually much
smaller.

At the end of Phase 2 we are left with a set of rows, each one representing a
position in $T$ where an order-preserving match can occur. We verify if there
is actually a match ({\sf Phase 3}) by decompressing the corresponding
subsequence of $T$ and comparing it with $P$. Given a row index $r$
representing a position in $T_o$ prefixed by a string $x_2  x_3 \cdots
x_{q-1} P[q]\cdots P[p]$ we use the LF-map to move backwards in $T_o$ until we
reach a marked position, that is, a position in $T_o$ (and hence in $T$)
which is a multiple of the block size $B$ (say position $\ell B$) and marks
the beginning of block~$\ell$. Each time we apply the LF-map we also obtain a
symbol $y_i$ of $T_o$ hence when we reach the beginning of the block we also
have the sequence
$$
y_1\, y_2\, \cdots y_k\, x_2\, x_3\, \cdots x_{q-1} P[q]\cdots P[p]
$$
of $T_o$ values from the beginning of the block till the position
corresponding to $P[p]$. Using this information and the compressed
representation of $T_\delta$ (whose blocks can be accessed independently) we
retrieve the corresponding $T$ values and determine if there is an actual
order-preserving match.

Phase 3 is usually the most expensive since for each candidate the algorithm
has to reach the beginning of the block containing it. We can therefore
expect that its running time will be linearly affected by the block size $B$. Note
that in our implementation Phase 2 and 3 are interleaved: as soon as we have
determined a range of rows prefixed by one of the patterns in
Corollary~\ref{cor:main} we execute Phase 3 for all rows in the range before
considering any other row range.

\begin{figure}
\begin{center}\setlength{\tabcolsep}{7pt}
\begin{tabularx}{0.9\textwidth}{lr l @{\extracolsep{\fill} }}
\hline
{\sf Name}   & \multicolumn{1}{r}{{\sf \# Values}} & {\sf Description} \\\hline
{\sf prices} & 31,559,990 & daily, hourly, and 5min US stock prices\\
{\sf temp}   & 30,505,702 & max and min daily temperature from 424 US stations\\
{\sf ecg}    & 20,138,750 & 22 hours and 23 minutes of ECG data\\
{\sf rwalk}  & 50,000,000 & random walk with integer steps in the range $[-20,20]$\\
{\sf rand}   & 50,000,000 & random integers in the range $[-20,20]$\\
{\sf ran127} & 50,000,000 & random integers in the range $[-127,127]$\\
\hline
\end{tabularx}
\end{center}\caption{Files used in our experiments. All values are 32-bit integers
so the size in bytes of the files is four times the number of values.\label{fig:files}}
\end{figure}

\section{Experimental results} \label{sec:experiments}

Since one of the applications of order-preserving matching is the search of trends in Stock Market data we call our tool ``stock market index'' (\smi\ from now on).  Ours is the first tool combining compression and indexing for the order-preserving matching problem, so we have no direct competitors. We therefore consider separately the compression and search capabilities of our index. We compare the compression ratio achieved by \smi\ with those of \gzip\ and \xz: the former has been the standard compression tool for more than 20 years, while the latter is a more recent compressor based on the Lempel-Ziv-Markov chain algorithm (LZMA) that uses more resources but achieves a significantly better compression. 

We compare \smi's search speed with those of three different tools. As a baseline we used \scan,  a naive algorithm that simply tries to match the pattern in every text position using the verification algorithm outlined in~\cite[Sec.~3]{ipl/ChhabraT16}. Then, we tested a prototype of the algorithm \fmo from~\cite{stringology/ChhabraKT15} that is similar to our approach in that it builds an index to speedup the search. \fmo builds an FM-index of the binary string obtained by replacing each text character with 0 or 1 according to whether the next character is smaller than the current one or not. To search for a pattern \fmo applies the same transformation to the pattern characters and uses the index to quickly determine a set of positions where the pattern may occur. To discard false positives \fmo does a final verification step using the original text. \ignore{similar to the one in our algorithm.}  \ignore{Currently \fmo builds the index at every invocation although it could be easily adapted to store the index alongside the text for successive invocations.}

Finally, we compare \smi\ with the algorithm~\sso~\cite{spe/TSOT16}. \sso does not use an index so it finds order-preserving matchings by testing every text position. However, to verify the matching \sso exploits in a clever way  the SIMD parallelism offered by the SSE instruction set; as a result \sso is in practice  the fastest among the order-preserving matching algorithms that do not use an index~\cite[Section~4]{spe/TSOT16}.

We tested all algorithms using the collection of files listed in Fig.~\ref{fig:files} which includes real and synthetic data. The test files, as well as the source code of all tested algorithms, are available at the repository~\url{https://gitlab.com/manzai/order-preserving-index}. All tests have been performed on a desktop PC with eight Intel-I7 3.40GHz CPUs running Linux-Debian 8.3 using the {\sf gcc} compiler version 4.9.2 and optimization option {\sf -O3}. All tests used a single CPU while the PC was not performing any other significant computation.




\begin{table}[p]
\begin{center}
\begin{tabularx}{0.95\textwidth}{|l|@{\extracolsep{\fill} }l|R|R|R|R|R|R|}
\hline
\multicolumn{2}{|c|}{}     & {prices}& {temp}  & ecg   & rwalk & rand  & ran127 \\\hline\hline
&\smi  $q\!=\!3 $ & 0.35  & 0.31  & 0.27  & 0.34  & 0.35  & 0.43 \\\cline{2-8}
&\smi  $q\!=\!6 $ & 0.37  & 0.33  & 0.29  & 0.35  & 0.35  & 0.43 \\\cline{2-8}$B=32$
&\smi  $q\!=\!9 $ & 0.38  & 0.34  & 0.30  & 0.37  & 0.36  & 0.44 \\\cline{2-8}
&\smi  $q\!=\!12$ & 0.39  & 0.35  & 0.31  & 0.38  & 0.37  & 0.45 \\\hline\hline
&\smi  $q\!=\!3 $ & 0.29  & 0.25  & 0.21  & 0.27  & 0.28  & 0.37 \\\cline{2-8}
&\smi  $q\!=\!6 $ & 0.30  & 0.26  & 0.22  & 0.28  & 0.28  & 0.37 \\\cline{2-8}$B=64$
&\smi  $q\!=\!9 $ & 0.31  & 0.27  & 0.23  & 0.30  & 0.29  & 0.37 \\\cline{2-8}
&\smi  $q\!=\!12$ & 0.32  & 0.28  & 0.24  & 0.30  & 0.30  & 0.37 \\\hline\hline
&\smi  $q\!=\!3 $ & 0.27  & 0.23  & 0.19  & 0.24  & 0.26  & 0.34 \\\cline{2-8}
&\smi  $q\!=\!6 $ & 0.28  & 0.24  & 0.20  & 0.26  & 0.26  & 0.34 \\\cline{2-8}$B=96$
&\smi  $q\!=\!9 $ & 0.29  & 0.25  & 0.21  & 0.27  & 0.27  & 0.35 \\\cline{2-8}
&\smi  $q\!=\!12$ & 0.30  & 0.26  & 0.22  & 0.28  & 0.27  & 0.35 \\\hline\hline
&\gzip \textsf{-{}-best}  & 0.37  & 0.24  & 0.19  & 0.36  & 0.24  & 0.35 \\\cline{2-8}
&\xz   \textsf{-{}-best}  & 0.24  & 0.17  & 0.12  & 0.23  & 0.18  & 0.29 \\\hline\hline
\end{tabularx}
\end{center}\caption{\smi overall space usage for different values of $B$ and $q$. Each value is the ratio between the size of the index over the size of the input file, both expressed in bytes. Also shown are the space usage for \gzip\ and \xz.\label{tab:space_strict}}
\end{table}




\begin{table}[p]
\begin{center}
\begin{tabularx}{0.95\textwidth}{|l|l|l|R|R|R|R|R|R|}
\hline
\multicolumn{3}{|c|}{} & {prices}& {temp}  & ecg   & rwalk & rand  & ran127\\\hline\hline
&$q\!=\!3 $&\csa~~~~& 0.10  & 0.10  & 0.10  & 0.10  & 0.10  & 0.10  \\\cline{3-9}
&          &\lsk    & 0.25  & 0.21  & 0.17  & 0.24  & 0.25  & 0.33  \\\cline{2-9}$B=32$
&$q\!=\!6 $&\csa    & 0.14  & 0.14  & 0.13  & 0.13  & 0.14  & 0.13  \\\cline{3-9}
&          &\lsk    & 0.23  & 0.19  & 0.16  & 0.22  & 0.21  & 0.30  \\\cline{2-9}
&$q\!=\!9$ &\csa    & 0.16  & 0.16  & 0.15  & 0.15  & 0.17  & 0.15  \\\cline{3-9}
&          &\lsk    & 0.22  & 0.18  & 0.15  & 0.21  & 0.20  & 0.29  \\\hline\hline
&$q\!=\!3$ &\csa   & 0.09  & 0.09  & 0.09  & 0.08  & 0.09  & 0.08  \\\cline{3-9}
&          &\lsk   & 0.20  & 0.16  & 0.12  & 0.18  & 0.19  & 0.28  \\\cline{2-9}$B=64$
&$q\!=\!6$ &\csa   & 0.12  & 0.13  & 0.12  & 0.12  & 0.12  & 0.12  \\\cline{3-9}
&          &\lsk   & 0.18  & 0.13  & 0.11  & 0.17  & 0.16  & 0.25  \\\cline{2-9}
&$q\!=\!9$ &\csa   & 0.14  & 0.15  & 0.13  & 0.14  & 0.15  & 0.14  \\\cline{3-9}
&          &\lsk   & 0.17  & 0.12  & 0.10  & 0.16  & 0.14  & 0.23  \\\hline\hline
\end{tabularx}
\end{center}\caption{Space usage of $T_o$'s compressed
suffix array (\csa) and of $\Tdelta$'s log-skewed encoding (\lsk) for
different values of $B$ and $q$. The reported values are the ratio between
the size of the \csa/\lsk\ file over the size of the input file both expressed in
bytes.\label{tab:compr1_strict}}
\end{table}

\subsection{Space usage}\label{sec:exp:space}

Table~\ref{tab:space_strict} shows the space usage of our index for different
values of $B$ and $q$ compared to the compression tools \gzip\ and~\xz. We can see that \smi's space usage is essentially at par with
\gzip's: it can be smaller or larger depending on the block size $B$. As
expected \xz\ compression is clearly superior to both.
Table~\ref{tab:compr1_strict} shows the relative space usage, within \smi, of the compressed suffic array (\csa) for $T_o$ vs. the long-skewed encoding (\lsk) of $\Tdelta$. For a fixed block size $B$, as the window size $q$ increases, the cost of storing $\Tdelta$ decreases while the \csa\ size
increases. This was to be expected since a larger $q$ means that more
information is contained in $T_o$. For a fixed window size $q$, as the block
size $B$ increases, the space of both $\Tdelta$ and $T_o$'s \csa\ decreases
since both structures have an extra overhead for each block. However,
increasing the block size decreases the search speed as discussed in the
following section.

\subsection{Search time}\label{sec:exp:time}

All search tests involved 1000 patterns of length 10, 15 and 20 extracted from the same file where the patterns are later sought, so every search reports at least one occurrence. The patterns were extracted selecting 1000 random positions in the file. Note that patterns occurring more often are more likely to be selected so this setting is the least favorable for our algorithm: like all index-based algorithms, it is much faster when the pattern does not occur, or occurs relatively few times.

Since Phases 1 and 2 of our algorithm produce a set of candidates that must
be verified in Phase 3, in our first experiment we measure how effective are
the first two Phases in producing only a small number of candidates which are
later discarded (that is, how effective are Phases~1 and~2 in producing a
small number of false positives). The results of this experiment are reported
in Table~\ref{tab:fp_strict}. Note that the number of false positives does not depend on the block size $B$ that only affects the space usage and running time.


We see that for patterns of length $10$ the number of false positives can be very high especially for $q=3$ and $q=9$. This is due to two different phenomena. For $q=3$, the order component $T_o$ has the least amount of information on the actual ordering, so even after Phase 2 the number of false positives is significant. For $q=9$, $T_o$ is much more informative but during Phase~1 we search in $T_o$ \csa\ only the two symbols $P_o[9,10]$ (see Section~\ref{sec:search}) and therefore we are left with a large number of candidates. In Phase~2 however, we make use of $P_o[2,8]$ and the number of false positives drops below those of the case $q=3$.

As the size of the pattern increases, we see that the number of false positives decreases significantly. For $|P|=20$ at the end of Phase~2 the number of false positives is at most 64\% of the number of occurrences for $q=3$ and such percentage drops to $5\%$ for $q=6$ and $1\%$ for $q=9$. The bottom line is that the number of false positives appears to be reasonably small except in extreme cases when $q$ is very small or when $q$ is very close to the length of the pattern. Note also that there is a large variability in the number of false positives across the different input files.





\begin{table}[t]\renewcommand{\arraystretch}{1}
\begin{center}

\begin{tabularx}{0.95\textwidth}{|l|@{\extracolsep{\fill} }l |R|R|R|R|R|R|}
\hline \multicolumn{2}{|l|}{$|P|=10$}
                & {prices}& {temp}  & ecg   & rwalk & rand  & ran127 
                \\\hline\hline\multicolumn{2}{|l|}{ave \# occs}
                & 431.45 &   8.42 &19801.90 &802.39 &   3.03 &  10.66 \\\hline\hline
 $q=3$& Phase~1 &   8.02 & 311.89 &   1.76 &  12.81 & 2984.41 & 1470.69 \\\cline{2-8}
      & Phase~2 &   4.81 & 182.44 &   1.19 &   8.13 & 2035.82 & 1033.44 \\\hline\hline
 $q=6$& Phase~1 &  26.31 & 274.70 &   5.47 &  35.30 & 808.16  & 563.43  \\\cline{2-8}
      & Phase~2 &   1.83 &  10.11 &   1.04 &   2.81 &  53.33  &  53.21  \\\hline\hline
 $q=9$& Phase~1 &1193.64 &30119.77&  42.96 & 1325.32&104819.21&50448.11 \\\cline{2-8}
      & Phase~2 &   3.72 &  15.79 &   1.05 &   6.06 & 221.74  & 308.66  \\\hline
\end{tabularx}

\vspace{0.5cm}

\begin{tabularx}{0.95\textwidth}{|l|@{\extracolsep{\fill} }l |R|R|R|R|R|R|}
\hline \multicolumn{2}{|l|}{$|P|=15$}
                & {prices} & {temp}   & ecg    & rwalk  & rand   & ran127 \\\hline\hline\multicolumn{2}{|l|}{ave \# occs}
                & 59.70 &   1.00 & 2015.25&   2.38 &   1.00 &   1.00 \\\hline\hline
 $q=3$& Phase~1 &  1.37 &  10.10 &   2.01 &  23.62 &  49.70 & 119.48 \\\cline{2-8}
      & Phase~2 &  1.16 &   6.64 &   1.23 &  15.36 &  34.53 &  85.26 \\\hline\hline
 $q=6$& Phase~1 &  3.31 &   1.47 &   4.48 &  20.26 &   1.20 &   2.25 \\\cline{2-8}
      & Phase~2 &  1.02 &   1.02 &   1.03 &   2.21 &   1.01 &   1.11 \\\hline\hline
 $q=9$& Phase~1 & 11.04 &   5.39 &  13.03 & 438.92 &   2.11 &   9.19 \\\cline{2-8}
      & Phase~2 &  1.04 &   1.00 &   1.03 &   3.24 &   1.00 &   1.06 \\\hline
\end{tabularx}

\vspace{0.5cm}

\begin{tabularx}{0.95\textwidth}{|l|@{\extracolsep{\fill} }l |R|R|R|R|R|R|}
\hline \multicolumn{2}{|l|}{$|P|=20$}
                & {prices} & {temp}   & ecg    & rwalk  & rand   & ran127 \\\hline\hline\multicolumn{2}{|l|}{ave \# occs}
                &   1.49 &   1.00 & 147.80 &   1.00 &   1.00 &   1.00 \\\hline\hline
 $q=3$& Phase~1 &   1.20 &   1.05 &   2.60 &   1.39 &   1.27 &   1.89 \\\cline{2-8}
      & Phase~2 &   1.12 &   1.03 &   1.42 &   1.29 &   1.17 &   1.64 \\\hline\hline
 $q=6$& Phase~1 &   2.49 &   1.00 &   5.34 &   1.08 &   1.00 &   1.00 \\\cline{2-8}
      & Phase~2 &   1.00 &   1.00 &   1.05 &   1.02 &   1.00 &   1.00 \\\hline\hline
 $q=9$& Phase~1 &  47.19 &   1.00 &  17.17 &   2.10 &   1.00 &   1.00 \\\cline{2-8}
      & Phase~2 &   1.01 &   1.00 &   1.02 &   1.01 &   1.00 &   1.00 \\\hline
\end{tabularx}

\end{center}\caption{False positives as a function of the window size~$q$ for 1000 patterns of length 10, 15 and 20. The first row shows the average number
of actual occurrences for the patterns in the test set. The other rows show
the ratios between candidates and actual occurrences at the end of Phase 1
and 2, obtained as the sum of candidates over all patterns over the sum of actual occurrences over all patterns.\label{tab:fp_strict}}
\end{table}




\begin{table}[t] 
\begin{center}\renewcommand{\arraystretch}{1}
\begin{tabularx}{0.95\textwidth}{|l|l|R|R|R|R|R|R|R|}
\hline 
\multicolumn{2}{|l|}{$|P|=10$}     
                & {prices} & {temp}   & ecg    & rwalk  & rand   & ran127 \\\hline\hline
\multicolumn{2}{|l|}{ave \#  occ} 
                & 431.45 &   8.42 &19801.90&802.39  &   3.03 &  10.66\\\hline\hline
&\smi $q\!=\!3$ &   7.52 &   6.16 &  64.82 &  27.17 &  26.96 &  46.27 \\\cline{2-8} $B\!=\!32$
&\smi $q\!=\!6$ &   4.15 &   0.62 &  75.82 &  15.01 &   1.35 &   4.28 \\\cline{2-8}
&\smi $q\!=\!9$ &  13.81 &   2.13 &  91.06 &  47.47 &  10.55 &  40.20 \\\hline\hline
&\smi $q\!=\!3$ &  13.68 &  10.89 & 125.11 &  50.37 &  50.27 &  87.06 \\\cline{2-8} $B\!=\!64$
&\smi $q\!=\!6$ &   7.59 &   1.09 & 151.87 &  28.12 &   2.40 &   7.75 \\\cline{2-8}
&\smi $q\!=\!9$ &  23.40 &   3.07 & 182.84 &  83.77 &  16.16 &  65.70 \\\hline\hline\multicolumn{2}{|l|}{\fmo}
                &  70.42 &  74.74 & 256.85 & 185.07 & 399.92 & 398.79 \\\hline\hline\multicolumn{2}{|l|}{\scan}
                & 246.19 & 239.98 & 134.51 & 408.60 & 420.18 & 442.69 \\\hline
\end{tabularx}

\vspace{0.5cm}


\begin{tabularx}{0.95\textwidth}{|l|l|R|R|R|R|R|R|R|}
\hline 
\multicolumn{2}{|l|}{$|P|=15$}     
                & {prices} & {temp}   & ecg    & rwalk  & rand   & ran127 \\\hline\hline
\multicolumn{2}{|l|}{ave \#  occ} 
                &  59.70 &   1.00 & 2015.25&   2.38 &   1.00 &   1.00 \\\hline\hline
&\smi $q\!=\!3$ &   0.19 &   0.03 &   6.96 &   0.17 &   0.17 &   0.40 \\\cline{2-8} $B\!=\!32$
&\smi $q\!=\!6$ &   0.19 &   0.02 &   7.58 &   0.05 &   0.02 &   0.02 \\\cline{2-8}
&\smi $q\!=\!9$ &   0.21 &   0.02 &   8.83 &   0.13 &   0.03 &   0.03 \\\hline\hline
&\smi $q\!=\!3$ &   0.35 &   0.06 &  13.28 &   0.30 &   0.30 &   0.71 \\\cline{2-8} $B\!=\!64$
&\smi $q\!=\!6$ &   0.37 &   0.02 &  15.13 &   0.08 &   0.02 &   0.03 \\\cline{2-8}
&\smi $q\!=\!9$ &   0.39 &   0.03 &  17.90 &   0.19 &   0.04 &   0.04 \\\hline\hline\multicolumn{2}{|l|}{\fmo}
                &   2.52 &   2.60 &  50.82 &   6.18 &  21.05 &  21.74 \\\hline\hline\multicolumn{2}{|l|}{\scan}
                & 248.71 & 240.32 & 134.43 & 403.49 & 406.74 & 454.96 \\\hline
\end{tabularx}

\vspace{0.5cm}

\begin{tabularx}{0.95\textwidth}{|l|l|R|R|R|R|R|R|R|}
\hline 
\multicolumn{2}{|l|}{$|P|=20$}
                & {prices} & {temp}   & ecg    & rwalk  & rand   & ran127 \\\hline\hline
\multicolumn{2}{|l|}{ave \#  occ} 
                &   1.49 &   1.00 & 147.80 &   1.00 &   1.00 &   1.00 \\\hline\hline
&\smi $q\!=\!3$ &   0.01 &   0.01 &   0.66 &   0.01 &   0.01 &   0.01 \\\cline{2-8} $B\!=\!32$
&\smi $q\!=\!6$ &   0.02 &   0.02 &   0.63 &   0.02 &   0.02 &   0.02 \\\cline{2-8}
&\smi $q\!=\!9$ &   0.03 &   0.03 &   0.74 &   0.03 &   0.03 &   0.03 \\\hline\hline
&\smi $q\!=\!3$ &   0.02 &   0.01 &   1.16 &   0.02 &   0.02 &   0.02 \\\cline{2-8} $B\!=\!64$
&\smi $q\!=\!6$ &   0.03 &   0.03 &   1.19 &   0.03 &   0.03 &   0.03 \\\cline{2-8}
&\smi $q\!=\!9$ &   0.04 &   0.03 &   1.41 &   0.04 &   0.04 &   0.04 \\\hline\hline\multicolumn{2}{|l|}{\fmo}
                &   0.12 &   0.13 &   8.86 &   0.23 &   1.18 &   1.31 \\\hline\hline\multicolumn{2}{|l|}{\scan}
                & 241.21 & 234.43 & 133.22 & 402.90 & 405.34 & 468.98 \\\hline
\end{tabularx}

\end{center}\caption{Average \smi running times, in milliseconds,
for searching 1000 random patterns of length 10, 15, and 20 for different values of $B$ and $q$ compared with the running times of the algorithms \scan and \fmo. Running times do not
include the time to load/compute the compressed index (for \ssmi and \fmo) or the uncompressed
text (for \scan and \fmo).\label{tab:speed}}
\end{table}

Table~\ref{tab:speed} reports the average running times of \smi\ for different values of $q$ and $B$, compared with those of the naive algorithm \scan and of the algorithm \fmo from~\cite{stringology/ChhabraKT15} (see description at the beginning of Section~\ref{sec:experiments}). We see that our algorithm is the fastest for every instance and the difference is by at least an order of magnitude for the larger pattern lengths and $B=32$. If we look at the running times for pattern lengths 15 and 20 we see that, as expected for an index data structure, the running time grows linearly with the number of occurrences and does not depend on the size of the input. For example, the search is much faster for {\sf rand} than for {\sf ecg} since the latter is much shorter but has a much larger number of occurrences. For patterns of length 10 we see that the running time can be influenced by the presence of a large number of false positives at the end of Phase 2. For example, the two files {\sf rwalk} and {\sf ran127} have the same length; the former has a much higher number of occurrences, but the search times are relatively close since the latter has a much higher number of false positives. Finally, we observe that doubling the block size from $B=32$ to $B=64$ roughly doubles the search time for all values of $q$. Since the size of the block influences the running time only in Phase 3, when we decompress portions of $\Tdelta$, this is an indirect confirmation that Phase~3 is indeed the most expensive of the algorithm. Since increasing the block size from $B=32$ to $B=64$ improves compression only slightly (see Table~\ref{tab:compr1_strict}), the lesson we learn here is that larger blocks should be used only when space is at a premium.

The algorithm \fmo is also based on an index: before the search \fmo builds an FM-index on the binary string encoding whether the next character is smaller than the current one or not. Therefore, in some sense \fmo search operations are similar to the ones we would have for \smi with $q=2$ (in our implementation we can only have $3 \leq q \leq 128$ so we could not test this). Note however that we also keep explicit track of equalities of text values, so with $q=3$ each entry in $T_o$ can assume five distinct values instead of two as in \fmo. For $|P|=10$ we see that \fmo's running times are between 4 and 14 times \smi's corresponding running times for $q=3$ and $B=32$. For larger values of $|P|$ the gap between \fmo and \smi running times increases. With some additional tests, not reported here, we found that the likely reason is that even for $|P|=20$ the search in \fmo's index still produces a large number of false positives that have to be checked and discarded.

As expected, \scan running time is roughly proportional to the input file length and is scarcely affected by the pattern length and the number of occurrences. Note that for $P=15$ and $P=20$ \scan is at least two orders of magnitude slower than \smi\ for $B=32$, with the only exception of the input file {\sf ecg} and $|P|=15$ which have an unusually high average number of occurrences (more than 2000).

In Table~\ref{tab:speed_sso} we report a comparison including the algorithm \sso~\cite{spe/TSOT16}, which is the fastest algorithm among those that do not use an index. The available version of \sso is optimized for small values so it only accepts input files with values in the range $[-127,127]$. In our datasets of Table~\ref{fig:files} the files {\sf temp}, {\sf rand} and {\sf ran127} already satisfy this restriction. To include also the other files we changed them by transforming each value $x$ to $(x \bmod 255) - 127$. With this transformation all values are forced into the range $[-127,127]$ and, if neighboring values are not too distant, their relative order is usually unchanged. In Table~\ref{tab:speed_sso} the names of the files whose values has been modified with the above transformation are shown in boldface.

From the results in Table~\ref{tab:speed_sso} we see that \sso running time is proportional to the input file length and is roughly thirty times faster than \scan. Indeed, for $|P|=10$ \sso is the fastest algorithm for the (modified) {\sf ecg} file and is faster than \smi for $q=3$ also for {\sf randw}, {\sf rand} and {\sf ran127}. However, for larger values of the pattern length \smi becomes significantly faster than \sso. We also notice that for $|P|=15$ there is not a clear winner between \fmo and \sso, while for $|P|=20$ \fmo is consistently faster.

Summing up, we believe that our experiments indicate that for sufficiently large files it pays to build an index also for the order-preserving matching problem. Compared to the exact matching problem, the index performance for order-preserving matching is less predictable because the proposed indices, \smi and \fmo, are only ``approximate'' and must deal with the possible presence of false positives. Nevertheless, the advantages of using an index are clear, especially when the index can be compressed to use significantly less space than the original file, as in \smi. On the other hand, carefully engineered scan-based algorithms, like \sso, will always be needed for short-medium files, or for very short patterns, or for the case in which the text cannot be preprocessed in advance.




\begin{table} 
\begin{center}
\begin{tabularx}{0.95\textwidth}{|l|R|R|R|R|R|R|}
\hline 
$|P|=10$       & \prcz  & {temp} & \ecgz  & \randwz& rand   & ran127   \\\hline\hline
ave \#  occ    & 358.22 &   8.42 & 7460.50& 622.32 &   3.03 &  10.66 \\\hline\hline
\smi $q\!=\!3$ &   6.37 &   6.16 &  27.76 &  22.42 &  26.96 &  46.27 \\\hline
\smi $q\!=\!6$ &   2.19 &   0.62 &  31.73 &  10.97 &   1.35 &   4.28 \\\hline
\smi $q\!=\!9$ &   9.30 &   2.13 &  40.09 &  42.38 &  10.55 &  40.20 \\\hline\hline
\sso           &   7.41 &   6.70 &   5.36 &  14.14 &  11.81 &  14.08 \\\hline\hline
\fmo           &  74.76 &  71.88 & 133.73 & 187.01 & 397.84 & 397.93 \\\hline\hline
\scan          & 254.01 & 239.98 & 139.96 & 417.25 & 420.18 & 442.69 \\\hline
\end{tabularx}

\vspace{0.5cm}

\begin{tabularx}{0.95\textwidth}{|l|R|R|R|R|R|R|}
\hline 
$|P|=15$       & \prcz   & {temp}& \ecgz  & \randwz& rand   & ran127   \\\hline\hline
ave \#  occ    & 169.93 &   1.00 & 880.49 &   1.22 &   1.00 &   1.00 \\\hline\hline
\smi $q\!=\!3$ &   0.42 &   0.03 &   3.17 &   0.12 &   0.17 &   0.40 \\\hline
\smi $q\!=\!6$ &   0.45 &   0.02 &   3.79 &   0.03 &   0.02 &   0.02 \\\hline
\smi $q\!=\!9$ &   0.50 &   0.02 &   4.47 &   0.07 &   0.03 &   0.03 \\\hline\hline
\sso           &   7.37 &   6.41 &   5.28 &  13.58 &  10.92 &  13.82 \\\hline\hline
\fmo           &   2.93 &   2.57 &  17.05 &   6.47 &  21.00 &  21.63 \\\hline\hline
\scan          & 249.15 & 240.32 & 139.32 & 415.13 & 406.74 & 454.96 \\\hline
\end{tabularx}

\vspace{0.5cm}

\begin{tabularx}{0.95\textwidth}{|l|R|R|R|R|R|R|}
\hline 
$|P|=20$       & \prcz  & {temp} & \ecgz  & \randwz& rand   & ran127   \\\hline\hline
ave \#  occ    &   1.08 &   1.00 &  42.50 &   1.00 &   1.00 &   1.00 \\\hline\hline
\smi $q\!=\!3$ &   0.01 &   0.01 &   0.22 &   0.01 &   0.01 &   0.01 \\\hline
\smi $q\!=\!6$ &   0.02 &   0.02 &   0.21 &   0.02 &   0.02 &   0.02 \\\hline
\smi $q\!=\!9$ &   0.03 &   0.03 &   0.26 &   0.03 &   0.03 &   0.03 \\\hline\hline
\sso           &   7.07 &   6.36 &   5.18 &  12.82 &  10.34 &  13.47 \\\hline\hline
\fmo           &   0.15 &   0.12 &   2.20 &   0.26 &   1.26 &   1.38 \\\hline\hline
\scan          & 247.43 & 234.43 & 137.68 & 402.10 & 405.34 & 468.98 \\\hline

\end{tabularx}
\end{center}\caption{Average \smi running times for searching 1000 random patterns of length 10, 15, and 20 for different values of $q$ and $B=32$ compared with the running times of the algorithms \sso, \fmo, and \scan. Running times are in milliseconds and do not include the time to load/compute the compressed index (for \ssmi and \fmo) or the uncompressed text (for \sso, \fmo, and \scan). The names of the files whose content has been modified to force all values in the range $[-127,127]$ is in boldface.\label{tab:speed_sso}}
\end{table}

\section{Concluding Remarks}

In this paper we have proposed a compressed index for the order-preserving
pattern matching problem. Our approach is based on the new idea of
splitting the original sequence into two complementary components: the order
component and the $\delta$ component. The problem of finding the order-preserving occurrences of a pattern is transformed into an exact search
problem on the order component followed by a verification phase using the
$\delta$ component. Experiments show that our index has a space usage similar
to \gzip\ and can find order-preserving occurrences much faster than a
sequential scan.

Our approach is quite general and improvements could be obtained by changing
some implementation choices. For example, we index the order component using
a Wavelet-Tree based FM-index; to improve the performances for inputs with
many (order-preserving) repetitions we can use a different compressed
full-text index, for example adapting the one recently proposed in~\cite{soda/GagieNP18}. Notice also that the compression of the $\delta$ component can be radically changed without altering the overall scheme.

Finally, we define the order component considering the position of
the predecessor of each element in a sliding window. It is natural to try to
extend the approach by also considering the successor. This can be done
representing each element with a $\langle$predecessor, successor$\rangle$
pair, or using a second index storing the successor information.

\section*{Acknowledgments}

The authors would like to than Jorma Tarhio and Tamanna Chhabra for providing the code of \sso and \fmo, and for assistance in the testing of those algorithms.

\bibliographystyle{plain}

\bibliography{stockm}

\begin{thebibliography}{10}

\bibitem{isaac/BelazzouguiPRV13}
Djamal Belazzougui, Adeline Pierrot, Mathieu Raffinot, and St{\'{e}}phane
  Vialette.
\newblock Single and multiple consecutive permutation motif search.
\newblock In {\em {ISAAC}}, volume 8283 of {\em Lecture Notes in Computer
  Science}, pages 66--77. Springer, 2013.

\bibitem{stringology/CantoneFK15}
Domenico Cantone, Simone Faro, and M.~Oguzhan K{\"{u}}lekci.
\newblock An efficient skip-search approach to the order-preserving pattern
  matching problem.
\newblock In {\em Stringology}, pages 22--35. Department of Theoretical
  Computer Science, Czech Technical University in Prague, 2015.

\bibitem{spe/TSOT16}
Tamanna Chhabra, Simone Faro, M.~Oguzhan K{\"{u}}lekci, and Jorma Tarhio.
\newblock Engineering order-preserving pattern matching with {SIMD}
  parallelism.
\newblock {\em Softw., Pract. Exper.}, 47(5):731--739, 2017.

\bibitem{spire/ChhabraGT15}
Tamanna Chhabra, Emanuele Giaquinta, and Jorma Tarhio.
\newblock Filtration algorithms for approximate order-preserving matching.
\newblock In {\em {SPIRE}}, volume 9309 of {\em Lecture Notes in Computer
  Science}, pages 177--187. Springer, 2015.

\bibitem{stringology/ChhabraKT15}
Tamanna Chhabra, M.~Oguzhan K{\"{u}}lekci, and Jorma Tarhio.
\newblock Alternative algorithms for order-preserving matching.
\newblock In {\em Stringology}, pages 36--46. Department of Theoretical
  Computer Science, Czech Technical University in Prague, 2015.

\bibitem{wea/ChhabraT14}
Tamanna Chhabra and Jorma Tarhio.
\newblock Order-preserving matching with filtration.
\newblock In {\em {SEA}}, volume 8504 of {\em Lecture Notes in Computer
  Science}, pages 307--314. Springer, 2014.

\bibitem{ipl/ChhabraT16}
Tamanna Chhabra and Jorma Tarhio.
\newblock A filtration method for order-preserving matching.
\newblock {\em Inf. Process. Lett.}, 116(2):71--74, 2016.

\bibitem{ipl/ChoNPS15}
Sukhyeun Cho, Joong~Chae Na, Kunsoo Park, and Jeong~Seop Sim.
\newblock A fast algorithm for order-preserving pattern matching.
\newblock {\em Inf. Process. Lett.}, 115(2):397--402, 2015.

\bibitem{spire/CrochemoreIKKLPRRW13}
Maxime Crochemore, Costas Iliopoulos, Tomasz Kociumaka, Marcin Kubica, Alessio
  Langiu, Solon Pissis, Jakub Radoszewski, Wojciech Rytter, and Tomasz Walen.
\newblock Order-preserving incomplete suffix trees and order-preserving
  indexes.
\newblock In {\em {SPIRE}}, volume 8214 of {\em Lecture Notes in Computer
  Science}, pages 84--95. Springer, 2013.

\bibitem{tcs/CrochemoreIKKLP16}
Maxime Crochemore, Costas~S. Iliopoulos, Tomasz Kociumaka, Marcin Kubica,
  Alessio Langiu, Solon~P. Pissis, Jakub Radoszewski, Wojciech Rytter, and
  Tomasz Walen.
\newblock Order-preserving indexing.
\newblock {\em Theor. Comput. Sci.}, 638:122--135, 2016.

\bibitem{dcc/DecaroliGM17}
Gianni Decaroli, Travis Gagie, and Giovanni Manzini.
\newblock A compact index for order-preserving pattern matching.
\newblock In {\em {DCC}}, pages 72--81. {IEEE}, 2017.

\bibitem{FaroK16}
Simone Faro and M.~Oguzhan K{\"{u}}lekci.
\newblock Efficient algorithms for the order preserving pattern matching
  problem.
\newblock In {\em {AAIM}}, volume 9778 of {\em Lecture Notes in Computer
  Science}, pages 185--196. Springer, 2016.

\bibitem{isci/FerraginaM01}
Paolo Ferragina and Giovanni Manzini.
\newblock An experimental study of a compressed index.
\newblock {\em Inf. Sci.}, 135(1-2):13--28, 2001.

\bibitem{FM05}
Paolo Ferragina and Giovanni Manzini.
\newblock Indexing compressed text.
\newblock {\em Journal of the ACM (JACM)}, 52(4):552--581, 2005.

\bibitem{esa/GagieMV17}
Travis Gagie, Giovanni Manzini, and Rossano Venturini.
\newblock An encoding for order-preserving matching.
\newblock In {\em {ESA}}, volume~87 of {\em LIPIcs}, pages 38:1--38:15. Schloss
  Dagstuhl - Leibniz-Zentrum fuer Informatik, 2017.

\bibitem{soda/GagieNP18}
Travis Gagie, Gonzalo Navarro, and Nicola Prezza.
\newblock Optimal-time text indexing in bwt-runs bounded space.
\newblock In {\em {SODA}}, pages 1459--1477. {SIAM}, 2018.

\bibitem{cpm/GawrychowskiU14}
Pawel Gawrychowski and Przemyslaw Uznanski.
\newblock Order-preserving pattern matching with $k$ mismatches.
\newblock In {\em {CPM}}, volume 8486 of {\em Lecture Notes in Computer
  Science}, pages 130--139. Springer, 2014.

\bibitem{wea/GogBMP14}
Simon Gog, Timo Beller, Alistair Moffat, and Matthias Petri.
\newblock From theory to practice: Plug and play with succinct data structures.
\newblock In {\em {SEA}}, volume 8504 of {\em Lecture Notes in Computer
  Science}, pages 326--337. Springer, 2014.

\bibitem{Gusfield:97}
D.~Gusfield.
\newblock {\em Algorithms on Strings, Trees, and Sequences: Computer Science
  and Computational Biology}.
\newblock Cambridge University Press, 1997.

\bibitem{tcs/KimEFHIPPT14}
Jinil Kim, Peter Eades, Rudolf Fleischer, Seok{-}Hee Hong, Costas~S.
  Iliopoulos, Kunsoo Park, Simon~J. Puglisi, and Takeshi Tokuyama.
\newblock Order-preserving matching.
\newblock {\em Theor. Comput. Sci.}, 525:68--79, 2014.

\bibitem{ipl/KubicaKRRW13}
Marcin Kubica, Tomasz Kulczynski, Jakub Radoszewski, Wojciech Rytter, and
  Tomasz Walen.
\newblock A linear time algorithm for consecutive permutation pattern matching.
\newblock {\em Inf. Process. Lett.}, 113(12):430--433, 2013.

\bibitem{spe/ManziniR04}
Giovanni Manzini and Marcella Rastero.
\newblock A simple and fast {DNA} compressor.
\newblock {\em Softw., Pract. Exper.}, 34(14):1397--1411, 2004.

\bibitem{csur/NavarroM07}
Gonzalo Navarro and Veli M{\"{a}}kinen.
\newblock Compressed full-text indexes.
\newblock {\em {ACM} Comput. Surv.}, 39(1), 2007.

\end{thebibliography}

\end{document}